\documentclass{llncs}
\usepackage{pgf,tikz}
\usetikzlibrary {arrows}
\usetikzlibrary {decorations}
\usetikzlibrary{snakes}
\usepackage{main}
\usepackage{amsmath,amssymb,wasysym,stmaryrd}
\usepackage{enumerate}
\usepackage{times}

\makeatletter
\let\c@definition\c@theorem
\let\c@lemma\c@theorem
\let\c@corollary\c@theorem
\let\c@remark\c@theorem
\let\c@example\c@theorem
\let\c@proposition\c@theorem

\author{Thomas~Brihaye\inst{1} 
  \and Julie~De~Pril\inst{1}
  \and Sven~Schewe\inst{2}
}

\usepackage{url}
\urldef{\mail}\path|%
       {thomas.brihaye, julie.depril}@umons.ac.be|

 \institute{University of Mons - UMONS\\
   Place du Parc 20, 7000 Mons, Belgium\\
   \mail
   \and
   University of Liverpool\\
   \url{sven.schewe@liverpool.ac.uk}
   } 

 \title{Multiplayer Cost Games with\\Simple Nash Equilibria}

\begin{document}
\maketitle

\begin{abstract}
Multiplayer games with selfish agents naturally occur in the design of
distributed and embedded systems.  As the goals of selfish agents are
usually neither equivalent nor antagonistic to each other, such games
are non zero-sum games.  We study such games and show that a large
class of these games, including games where the individual objectives
are mean- or discounted-payoff, or quantitative reachability, and show
that they do not only have a solution, but a \emph{simple} solution.
We establish the existence of Nash equilibria that are composed of $k$
memoryless strategies for each agent in a setting with $k$ agents, one
main and $k-1$ minor strategies.  The main strategy describes what
happens when all agents comply, whereas the minor strategies ensure
that all other agents immediately start to co-operate against the
agent who first deviates from the plan.  This simplicity is important,
as rational agents are an idealisation.  Realistically, agents have to
decide on their moves with very limited resources, and complicated
strategies that require exponential---or even
non-elementary---implementations cannot realistically be implemented.
The existence of simple strategies that we prove in this paper
therefore holds a promise of implementability.
\end{abstract}

\section{Introduction}

% \paragraph{General framework.}  

The construction of correct and efficient computer systems (both hard-
and software) is recognised to be an extremely difficult task.  Formal
methods have been exploited with some success in the design and
verification of such systems.  Mathematical logic, automata
theory~\cite{HU}, and model-checking~\cite{CGP00} have contributed
much to the success of formal methods in this field.  However,
traditional approaches aim at systems with qualitative specifications
like LTL, and rely on the fact that these specifications are either
satisfied or violated by the system.

Unfortunately, these techniques do not trivially extend to complex
systems, such as embedded or distributed systems.  A main reason for
this is that such systems often consist of multiple independent
components with individual objectives.  These components can be viewed
as selfish agents that may cooperate and compete at the same time. It
is difficult to model the interplay between these components with
traditional finite state machines, as they cannot reflect the
intricate quantitative valuation of an agent on how well he has met
his goal.  In particular, it is not realistic to assume that these
components are always cooperating to satisfy a common goal, as it is,
e.g., assumed in works that distinguish between an environment and a
system.  We argue that it is more realistic to assume that all
components act like selfish agents that try to achieve their own
objectives and are either unconcerned about the effect this has on the
other components or consider this effect to be secondary.  It is
indeed a recent trend to enhance the system models used in the
classical approach of verification by quantitative cost and gain
functions, and to exploit the well established game-theoretic
framework~\cite{Na50,OR94} for their formal analysis.

The first steps towards the extension of computational models with
concepts from classical game theory were taken by advancing from
boolean to general two-player zero-sum games played on
graphs~\cite{GTW02}.  Like their qualitative counter parts, those
games are adequate to model controller-environment interaction
problems~\cite{T95,Th08}.  As usual in control theory, one can
distinguish between moves of a control player, who plays actions to
control a system to meet a control objective, and an antagonistic
environment player.  In the classical setting, the control player has
a qualitative objective---he might, for example, try to enforce a
temporal specification---whereas the environment tries to prevent
this.  In the extension to quantitative games, the controller instead
tries to maximise its gain, while the environment tries to minimise
it.  This extension lifts the controller synthesis problem from a
constructive extension of a decision problem to a classical
optimisation problem.

However, this extension has not lifted the restriction to purely antagonist interactions between a controller and a hostile environment.
In order to study more complex systems with more than two components, and with objectives that are not necessarily antagonist, we resort to multiplayer non zero-sum games.
In this context, \emph{Nash equilibria}~\cite{Na50} take the place that winning and optimal strategies take in qualitative and quantitative two-player games zero-sum games, respectively.
Surprisingly, qualitative objectives have so far prevailed in the study of Nash equilibria for distributed systems.
% One example of this would be that a player seeks to reach a set of target states, others would be games where each player has the goal to satisfy a temporal specification.
However, we argue that Nash equilibria for selfish agents with quantitative objectives---such as reaching a set of target states quickly or with a minimal consumption of energy---are natural objectives that aught to be studied alongside (or instead of) traditional qualitative objectives.

Consequently, we study \emph{Nash equilibria} for \emph{multiplayer non zero-sum} games played on graphs with \emph{quantitative} objectives.
% This paper provides some new theoretical results in this research direction.

\smallskip
\emph{Our contribution.}\ 
In this paper, we study turn-based multiplayer non zero-sum games
played on finite graphs with quantitative objectives, expressed
through a cost function for each player (\emph{cost games}).  Each
cost function assigns, for every play of the game, a value that
represents the cost that is incurred for a player by this play.  Cost
functions allow to express classical quantitative objectives such as
\emph{quantitative reachability} (i.e., the player aims at reaching a
subset of states as soon as possible), or \emph{mean-payoff}
objectives.  In this framework, all players are supposed to be
rational: they want to minimise their own cost or, equivalently,
maximise their own gain.  This invites the use of Nash equilibria as
the adequate concept for cost games.

%\fbox{describe mean-payoff with a sentence?} \fbox{add references?}

Our results are twofold.  Firstly, we prove the \emph{existence} of
Nash equilibria for a large class of cost games that includes
quantitative reachability and mean-payoff objectives.  Secondly, we
study the complexity of these Nash equilibria in terms of the
\emph{memory} needed in the strategies of the individual players in
these Nash equilibria.  More precisely, we ensure existence of Nash
equilibria whose strategies only requires a number of memory states
that is \emph{linear} in the size of the game for a wide class of cost
games, including games with quantitative reachability and mean-payoff
objectives.

% In particular, for quantitative reachability games, we prove that
% there exists a Nash equilibrium with memory at most $n +|V|$, where
% $n$ is the number of players and $|V|$ is the number of vertices of
% the game graph.

The general philosophy of our work is as follows: we try to derive
existence of Nash equilibria in multiplayer non zero-sum quantitative
games (and characterization of their complexity) through determinacy
results (and characterization of the optimal strategies) of several
well-chosen two-player quantitative games derived from the multiplayer
game.  These ideas were already successfully exploited in the
qualitative framework~\cite{GU08}, and in the case of limit-average
objectives~\cite{TT98}.
%Notice that we can draw a parallel in proof techniques between our
%paper and \cite{GU08} (qualitative framework): in both cases, the
%existence of a Nash equilibrium in a multiplayer game is deduced from
%the determinacy of some two-player games derived from the multiplayer
%game.

\smallskip
\emph{Related work.}\ Several recent papers have considered
\emph{two-player zero-sum} games played on finite graphs with regular
objectives enriched by some \emph{quantitative} aspects. Let us
mention some of them: games with finitary objectives~\cite{CH06},
mean-payoff parity games~\cite{CHJ05}, games with prioritised
requirements~\cite{AKW08}, request-response games where the waiting
times between the requests and the responses are
minimized~\cite{HTW08,Z09}, games whose winning conditions are
expressed via quantitative languages~\cite{BCHJ09}, and recently,
cost-parity and cost-Streett games~\cite{FZ12}.

Other work concerns \emph{qualitative non zero-sum} games.
In~\cite{GU08}, general criteria ensuring existence of Nash equilibria
and subgame perfect equilibria (resp.\ secure equilibria) are provided
for multiplayer (resp.\ $2$-player) games, as well as complexity
results.  The complexity of Nash equilibria in multiplayer concurrent
games with B\"uchi objectives has been discussed in~\cite{BBMU11}.
\cite{BBM10} studies the existence of Nash equilibria for timed games
with qualitative reachability objectives

Finally, there is a series of recent results on the combination of
\emph{non zero-sum} aspects with \emph{quantitative objectives}.
In~\cite{BG09}, the authors study games played on graphs with terminal
vertices where quantitative payoffs are assigned to the players.
%These games may have cycles, but all infinite plays form a single
%outcome (like in chess where every infinite play is a draw).  That
%paper gives criteria that ensure the existence of Nash (and subgame
%perfect) equilibria in pure and memoryless strategies.
In~\cite{KLST12}, 
%the studied games are played on priced graphs
%similar to the ones considered in this paper.  However, their games
%are concurrent, and it was shown the Nash equilibria do not
%necessarily exist in such games.  The 
the authors provide an algorithm to
decide the existence of Nash equilibria for concurrent priced games with quantitative reachability objectives.
%which uses a B\"uchi
%automaton that accepts all Nash equilibria outcomes, and establish the
%complexity of this decision problem. is also studied.  
In \cite{PS09}, the authors prove existence of a Nash equilibrium in
Muller games on finite graphs where players have a preference ordering
on the sets of the Muller table.  
%They show that Nash equilibria
%always exist for such games, and that it is decidable whether there
%exists a subgame perfect equilibrium.  In both cases they give a
%procedure to compute an equilibrium strategy profile, provided that
%such a strategy profile exists.
Let us also notice that the existence of a Nash equilibrium in cost
games with quantitative reachability objectives we study in this paper
has already been established in \cite{BBD10}.  The new proves we
provide are simpler and significantly improve the complexity of the
strategies constructed from exponential to linear in the size of the
game.

\smallskip
\emph{Organization of the paper.}\
In Section~\ref{sec:back}, we present the model of multiplayer cost
games and define the problems we study.  The main results are given in
Section~\ref{sec:results}. Finally, in Section~\ref{sec:appl}, we
apply our general result on particular cost games with classical
objectives. Omitted proofs and additional materials can be found in
the Appendix.

\section{General Background}\label{sec:back}
In this section, we define our model of \emph{multiplayer cost game},
recall the concept of Nash equilibrium and state the problems we
study.

\begin{definition}\label{def:cost game}
  A \emph{multiplayer cost game} is a tuple $\mathcal{G} =
  (\Pi, V, (V_i)_{i \in \Pi}, E, (\cost_i)_{i \in \Pi})$ where
  %\vspace{-.1cm}
  \begin{itemize}
  \item[\textbullet] $\Pi$ is a finite set of \emph{players},
%   \item[\textbullet] $G = (V,(V_i)_{i \in \Pi},E)$ is a \emph{finite
%     game graph}, i.e. a finite directed graph, where $V$ is the set of
%     vertices, $(V_i)_{i \in \Pi}$ is a partition of $V$ such that
%     $V_i$ is the set of vertices controlled by player~$i$, and $E
%     \subseteq V \times V$ is the set of edges, and
  \item[\textbullet] $G = (V,E)$ is a \emph{finite
    directed graph} with vertices $V$ and edges $E
    \subseteq V \times V$,
  \item[\textbullet] $(V_i)_{i \in \Pi}$ is a partition of $V$ such that
    $V_i$ is the set of vertices controlled by player~$i$, and

  \item[\textbullet] $\cost_i: \plays \to \IR \cup
    \{+\infty,-\infty\}$ is the \emph{cost function} of player~$i$,
    where $\plays$ is the set of \emph{plays} in $\mathcal{G}$,
    i.e. the set of infinite paths through $G$. For every play~$\rho
    \in \plays$, the value $\cost_i(\rho)$ represents the amount that
    player~$i$ loses for this play.
  \end{itemize}
\end{definition}
Cost games are \emph{multiplayer turn-based quantitative non zero-sum}
games. We assume that the players are rational: they play in a way to
minimise their own~cost.

Note that minimising cost or maximising gain are
essentially\footnote{Sometimes the translation implies minor follow-up
  changes, e.g., the replacement of $\liminf$ by $\limsup$ and vice
  versa.} equivalent, as maximising the gain for player $i$ can be
modelled by using $\cost_i$ to be minus this gain and then minimising
the cost.  This is particularly important in cases where two players
have antagonistic goals, as it is the case in all two-player zero-sum
games.  To cover these cases without changing the setting, we
sometimes refer to maximisation in order to preserve the connection to
such games in the literature.

% we could have defined a cost game
% with players winning a certain amount for plays and trying to maximise
% their payoff, or with a mix of these two kinds of players. But without
% loss of generality, we choose the definition with only players losing a
% cost. If a player wins payoffs according to a function $f: \plays \to
% \IR \cup \{+\infty,-\infty\}$, one can consider the cost function
% $\cost= -f$ in our setting.

For the sake of simplicity, we assume that each vertex has at least
one outgoing edge. Moreover, it is sometimes convenient to specify an
initial vertex $v_0 \in V$ of the game. We then call the pair
$(\mathcal{G},v_0)$ an \emph{initialised multiplayer cost game}. This
game is played as follows.  First, a token is placed on the initial
vertex~$v_0$.  Whenever a token is on a vertex $v \in V_i$ controlled
by player~$i$, player~$i$ chooses one of the outgoing edges $(v,v')\in
E$ and moves the token along this edge to $v'$.  This way, the players
together determine an \emph{infinite} path through the graph $G$,
which we call a \emph{play}. Let us remind that $\plays$ is the set of
all plays in $\mathcal{G}$.

A \emph{history}~$h$ of $\mathcal{G}$ is a \emph{finite} path through
the graph~$G$. We denote by $\hist$ the set of histories of a game,
and by $\epsilon$ the empty history. In the sequel, we write
$h=h_0\ldots h_k$, where $h_0,\ldots,h_k \in V$ ($k \in \IN$), for a
history~$h$, and similarly, $\rho=\rho_0\rho_1\ldots$, where
$\rho_0,\rho_1,\ldots \in V$, for a play~$\rho$.  A \emph{prefix} of
length $n+1$ (for some $n \in \IN$) of a
play~$\rho=\rho_0\rho_1\ldots$ is the finite history $\rho_0 \ldots
\rho_n$. We denote this history by $\rho[0,n]$.

Given a history~$h=h_0\ldots h_k$ and a vertex $v$ such that $(h_k,v)
\in E$, we denote by $hv$ the history $h_0\ldots h_kv$. Moreover,
given a history~$h=h_0\ldots h_k$ and a play~$\rho=\rho_0\rho_1\ldots$
such that $(h_k,\rho_0) \in E$, we denote by $h\rho$ the play
$h_0\ldots h_k \rho_0\rho_1\ldots$.

The function~$\last$ (resp.\ $\first$) returns, for a given history~$h=
h_0\ldots h_k$, the last vertex~$h_k$ (resp.\ the first vertex~$h_0$)
of $h$. The function $\first$ naturally extends to plays.

A \emph{strategy} of player~$i$ in $\mathcal{G}$ is a
function~$\sigma: \hist \to V$ assigning to each history~$h \in \hist$
that ends in a vertex $\last(h) \in V_i$ controlled by player $i$, a
successor $v=\sigma(h)$ of $\last(h)$. That is,
$\big(\last(h),\sigma(h)\big) \in E$.  We say that a
play~$\rho=\rho_0\rho_1\ldots$ of $\mathcal{G}$ is \emph{consistent}
with a strategy~$\sigma$ of player~$i$ if
$\rho_{k+1}=\sigma(\rho_0\ldots\rho_k)$ for all $k \in \IN$ such that
$\rho_k \in V_i$.  A \emph{strategy profile} of $\mathcal{G}$ is a
tuple~$(\sigma_i)_{\ipi}$ of strategies, where $\sigma_i$ refers to a
strategy for player~$i$.  Given an initial vertex $v$, a strategy
profile determines the unique play of $(\mathcal{G},v)$ that is
consistent with all strategies~$\sigma_i$.  This play is called the
\emph{outcome} of $(\sigma_i)_{\ipi}$ and denoted by $\langle
(\sigma_i)_{\ipi} \rangle_v$.  We say that a player \emph{deviates}
from a strategy (resp. from a play) if he does not carefully follow
this strategy (resp. this play).

%http://www.automata.rwth-aachen.de/~thomas/docs/knawmasterclass.pdf
%http://www.logic.rwth-aachen.de/~ummels/diplom.pdf 

A \emph{finite strategy automaton} for player $i \in \Pi$ over a game
$\mathcal{G}=(\Pi, V, (V_i)_{i \in \Pi},\linebreak E, (\cost_i)_{i \in
  \Pi})$ is a Mealy automaton $\mathcal{A}_i=(M,m_0,V,\delta,\nu)$
where:
\begin{itemize}
\item $M$ is a non-empty, finite set of memory states,
\item $m_0 \in M$ is the initial memory state,
\item $\delta: M \times V \to M$ is the memory update function,
\item $\nu: M \times V_i \to V$ is the transition choice function,
  such that $(v,\nu(m,v)) \in E$ for all $m \in M$ and $v \in V_i$.
\end{itemize}
We can extend the memory update function $\delta$ to a function
$\delta^*: M \times \hist \to M$ defined by $\delta^*(m,\epsilon)=m$
and $\delta^*(m,hv)=\delta(\delta^*(m,h),v)$ for all $m \in M$ and $hv
\in \hist$.  The strategy $\sigma_{\mathcal{A}_i}$ computed by a
finite strategy automaton ${\mathcal{A}_i}$ is defined by
$\sigma_{\mathcal{A}_i}(hv)= \nu(\delta^*(m_0,h),v)$ for all $hv \in
\hist$ such that $v \in V_i$. We say that $\sigma$ is a
\emph{finite-memory strategy} if there exists\footnote{Note that there
  exist several finite strategy automata such that
  $\sigma=\sigma_{\mathcal{A}}$.} a finite strategy automaton
${\mathcal{A}}$ such that $\sigma=\sigma_{\mathcal{A}}$. Moreover, we
say that $\sigma=\sigma_{\mathcal{A}}$ has a memory of size at most
$|M|$, where $|M|$ is the number of states of $\mathcal{A}$. In
particular, if $|M|=1$, we say that $\sigma$ is a \emph{positional
  strategy} (the current vertex of the play determines the choice of
the next vertex). We call $(\sigma_i)_{i \in \Pi}$ a strategy profile
with memory $m$ if for all $i \in \Pi$, the strategy $\sigma_i$ has a
memory of size at most $m$. A strategy profile~$(\sigma_i)_{\ipi}$ is
called \emph{positional} or \emph{finite-memory} if each~$\sigma_i$ is
a positional or a finite-memory strategy, respectively.

We now define the notion of \emph{Nash equilibria} in this
quantitative framework.

\begin{definition}\label{def:ne}
  Given an initialised multiplayer cost game $(\mathcal{G},v_0)$, a
  strategy profile~$(\sigma_i)_{\ipi}$ is a \emph{Nash equilibrium}
  in $(\mathcal{G},v_0)$ if, for every player~$j \in \Pi$ and for every
  strategy~$\sigma_j'$ of player~$j$, we have:
  $$\cost_j(\rho) \leq \cost_j(\rho')$$ where $\rho = \langle
  (\sigma_i)_{\ipi} \rangle_{v_0}$ and $\rho' = \langle
  \sigma_j',\sigma_{i \in \Pi \setminus \{j\}} \rangle_{v_0}$.
\end{definition}
This definition means that, for all~$j \in \Pi$, player~$j$ has no
incentive to deviate from $\sigma_j$ since he cannot strictly decrease
his cost when using~$\sigma_j'$ instead of~$\sigma_j$. Keeping
notations of Definition~\ref{def:ne} in mind, a strategy~$\sigma_j'$
such that $\cost_j(\rho) > \cost_j(\rho')$ is called a
\emph{profitable deviation} for player~$j$ w.r.t.~$(\sigma_i)_{\ipi}$.

%\documentclass{article}
%\usepackage{tikz}
%\usetikzlibrary {arrows}
%\usepackage{amsmath,amssymb}
%\newcommand{\cost}{{\sf{Cost}}}
%\newcommand{\price}{{\pi}}
%\newcommand{\hist}{{\sf{Hist}}}
%\begin{document}

\begin{example}\label{ex:cost game}
Let $\mathcal{G} = (\Pi, V, V_1,V_2, E, \cost_1,\cost_2)$ be the
two-player cost game whose graph $G=(V,E)$ is depicted in
Figure~\ref{fig-ex}. The states of player~$1$ (resp. $2$) are
represented by circles (resp. squares)\footnote{We will keep this
  convention through the paper.}. Thus, according to
Figure~\ref{fig-ex}, $V_1 = \{A,C,D\}$ and $V_2=\{B\}$. In order to
define the cost functions of both players, we consider a price
function $\price: E \to \{1,2,3\}$, which assigns a price to each edge
of the graph. The price function\footnote{Note that we could have
  defined a different price function for each player. In this case,
  the edges of the graph would have been labelled by couples of
  numbers.}  $\price$ is as follows (see the numbers in
Figure~\ref{fig-ex}): $\price(A,B)=\price(B,A)= \price(B,C)=1$,
$\price(A,D)=2$ and $\price(C,B)=\price(D,B)=3$. The cost function
$\cost_1$ of player~$1$ expresses a \emph{quantitative reachability
  objective}: he wants to reach the vertex $C$ (shaded vertex) while
minimising the sum of prices up to this vertex. That is, for every
play $\rho=\rho_0\rho_1\ldots$ of $\mathcal{G}$:
\[  \cost_1(\rho) = \left\{
\begin{array}{ll}
  \sum_{i=1}^n \price(\rho_{i-1},\rho_i) & \mbox{
    if $n$ is the \emph{least} index s.t. $\rho_n =C$,}\\ +
  \infty & \mbox{ otherwise.}
\end{array}\right. \]
As for the cost function $\cost_2$ of player~$2$, it expresses a
\emph{mean-payoff objective}: the cost of a play is the long-run
average of the prices that appear along this play. Formally, for any
play $\rho=\rho_0\rho_1\ldots$ of $\mathcal{G}$:
$$\cost_2(\rho)=\displaystyle \limsup_{n \to +\infty} \frac{1}{n}\cdot
\sum_{i=1}^n \price(\rho_{i-1},\rho_i).$$ Each player aims at
minimising the cost incurred by the play. Let us insist on the fact
that the players of a cost game may have different kinds of cost
functions (as in this example).

  \begin{figure}[h!]
    \centering
    \begin{tikzpicture}[xscale=1,yscale=.5,>=stealth',shorten >=1pt]
      \everymath{\scriptstyle}

      \path (0,0) node[draw,circle,inner sep=2pt] (qA) {$A$};
      \path (2.5,0) node[draw,rectangle,inner sep=3.5pt] (qB) {$B$};
      \path (5,0) node[draw,circle,inner sep=2pt,fill=black!20!white] 
      (qC) {$C$};
      \path (1.25,-1.75) node[draw,circle,inner sep=2pt] (qD) {$D$};
      
      \path (qA) edge[->,bend left] node[pos=.5,above] {{$1$}} (qB);
      \path (qB) edge[->,bend left] node[pos=.5,below] {{$1$}} (qA);

      \path (qB) edge[->,bend left] node[pos=.5,above] {{$1$}} (qC);
      \path (qC) edge[->,bend left] node[pos=.5,below] {{$3$}} (qB);

      \path (qA) edge[->,bend right] node[pos=.5,below] {{$2$}} (qD);
      \path (qD) edge[->,bend right] node[pos=.5,below] {{$3$}} (qB);

    \end{tikzpicture} \caption{A two-player cost game $\mathcal{G}$.}  
    \label{fig-ex} 
  \end{figure}
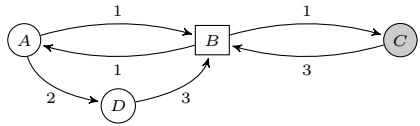
  
  An example of a play in~$\mathcal{G}$ can be given by $\rho =
  (AB)^\omega$, leading to the costs $\cost_1(\rho)= +\infty$ and
  $\cost_2(\rho)= 1$. In the same way, the play $\rho'=A(BC)^\omega$
  induces the following costs: $\cost_1(\rho)= 2$ and $\cost_2(\rho)=
  2$. 
  
  Let us fix the initial vertex~$v_0$ at the vertex~$A$. The
  play~$\rho = (AB)^\omega$ is the outcome of the positional
  strategy\footnote{Note that player~1 has no choice in vertices~$C$
    and~$D$, that is, $\sigma_1(hv)$ is necessarily equal to~$B$
    for~$v \in \{C,D\}$ and $h\in \hist$.}
  profile~$(\sigma_1,\sigma_2)$ where $\sigma_1(A)=B$ and
  $\sigma_2(B)=A$. Moreover, this strategy profile is in fact a
  \emph{Nash equilibrium}: player~$2$ gets the least cost he can
  expect in this game, and player~$1$ has no incentive to choose the
  edge $(A,D)$ (it does not allow the play to pass through
  vertex~$C$).

  We now consider the positional strategy profile $(\sigma_1',
  \sigma_2')$ with $\sigma_1'(A)=B$ and $\sigma_2'(B)=C$. Its outcome
  is the play $\rho'=A(BC)^\omega$. However, this strategy profile is
  \emph{not} a Nash equilibrium, because player~$2$ can strictly lower
  his cost by always choosing the edge $(B,A)$ instead of $(B,C)$,
  thus lowering his cost from 2 to 1. In other words, the strategy
  $\sigma_2$ (defined before) is a \emph{profitable deviation} for
  player~$2$ w.r.t.~$(\sigma_1',\sigma_2')$.

\end{example}
%\end{document}

The questions studied in this paper are the following ones:
\begin{pbm}\label{pbm 1}
  Given a multiplayer cost game $\mathcal{G}$, does there exist a Nash
  equilibrium in~$\mathcal{G}$?
\end{pbm}

\begin{pbm}\label{pbm 2}
  Given a multiplayer cost game $\mathcal{G}$, does there exist a
  finite-memory Nash equilibrium in~$\mathcal{G}$?
\end{pbm}

Obviously enough, if we make no restrictions on our cost games, the
answer to Problem~\ref{pbm 1} (and thus to Problem~\ref{pbm 2}) is
negative (see Example~\ref{NoNE}). Our first goal in this paper is to
identify a large class of cost games for which the answer to
Problem~\ref{pbm 1} is positive. Then we also positively reply to
Problem~\ref{pbm 2} for subclasses of the previously identified class
of cost games. Both results can be found in Section~\ref{sec:results}.

\begin{example}\label{NoNE}
Let $(\mathcal{G},A)$ be the initialised one-player cost game depicted
below, whose cost function $\cost_1$ is defined by
$\cost_1(A^nB^{\omega}) = \frac{1}{n}$ for $n \in \IN_0$ and
$\cost_1(A^\omega) = + \infty$. One can be convinced that there is no
Nash equilibrium in this initialised game.
%  \begin{figure}[h!]

    
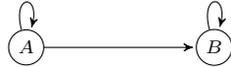
\centering
    \begin{tikzpicture}[xscale=1,yscale=1,>=stealth',shorten >=1pt]
      \everymath{\scriptstyle}

      \path (0,0) node[draw,circle,inner sep=2pt] (qA) {$A$};
      \path (2.5,0) node[draw,circle,inner sep=2pt] (qB) {$B$};
     
      \path (qA) edge[loop above] node[pos=.5,above] {{}} (qA);
      \path (qB) edge[loop above] node[pos=.5,above] {{}} (qB);
      \path (qA) edge[->] node[pos=.5,above] {{}} (qB);

    \end{tikzpicture}

%    \caption{A one-player cost game without a Nash equilibrium.}
%    \label{NoNE-fig}
%  \end{figure}
\end{example}

In order to our class of cost games, we need the notions of
\emph{Min-Max cost games}, \emph{determinacy} and \emph{optimal
  strategies}. The following two definitions are inspired
by~\cite{trivedi}.

\begin{definition}\label{def:min-max cost game}
  A \emph{Min-Max cost game} is a tuple $\mathcal{G} = (V,
  \VMin,\VMax,E,\costMin,\costMax)$, where
  \begin{itemize}
  \item[\textbullet] $G = (V,E)$ is a \emph{finite
    directed graph} with vertices $V$ and edges $E
    \subseteq V \times V$,
  \item[\textbullet] $(\VMin,\VMax)$ is a partition of $V$ such that
    $\VMin$ (resp. $\VMax$) is the set of vertices controlled by
    player~Min (resp. Max), and
  \item[\textbullet] $\costMin: \plays \to \IR \cup
    \{+\infty,-\infty\}$ is the \emph{cost function} of player~$\Min$,
    that represents the amount that he loses for a play, and
    $\costMax: \plays \to \IR \cup \{+\infty,-\infty\}$ is the
    \emph{gain function} of player~$\Max$, that represents the amount
    that he wins for a play.
  \end{itemize}
\end{definition}

In such a game, player $\Min$ wants to \emph{minimise} his cost, while
player $\Max$ wants to \emph{maximise} his gain. So, a Min-Max cost
game is a particular case of a two-player cost game. Let us stress
that, according to this definition, a Min-Max cost game is
\emph{zero-sum} if $\costMin=\costMax$, but this might not always be
the case\footnote{For an example, see the average-price game in
  Definition~\ref{def:part cost games}.}. We also point out that
Definition~\ref{def:min-max cost game} allows to take completely
unrelated functions $\costMin$ and $\costMax$, but usually they are
similar (see Definition~\ref{def:part cost games}). In the
sequel, we denote by $\SigmaMin$ (resp. $\SigmaMax$) the set of
strategies of player~$\Min$ (resp. $\Max$) in a Min-Max cost game.

\begin{definition}\label{def:determinacy}
  Given a Min-Max cost game $\mathcal{G}$, we define for every vertex
  $v \in V$ the \emph{upper value} $\valup(v)$ as:
  $$\valup(v)= \displaystyle \inf_{\sigma_1 \in \SigmaMin}
  \sup_{\sigma_2 \in \SigmaMax} \costMin(\langle \sigma_1,\sigma_2
  \rangle_v)\,,$$ and the \emph{lower value} $\vallow(v)$ as:
  $$\vallow(v)= \displaystyle \sup_{\sigma_2 \in \SigmaMax}
  \inf_{\sigma_1 \in \SigmaMin} \costMax(\langle \sigma_1,\sigma_2
  \rangle_v)\,.$$

  The game $\mathcal{G}$ is \emph{determined} if, for every $v \in V$,
  we have $\valup(v)=\vallow(v)$. In this case, we say that the game
  $\mathcal{G}$ has a \emph{value}, and for every $v \in V$, $\val(v)=
  \valup(v)= \vallow(v)$. We also say that the strategies $\sigmaopt_1
  \in \SigmaMin$ and $\sigmaopt_2 \in \SigmaMax$ are \emph{optimal
    strategies} for the respective players if, for every $v \in V$, we
  have that
  $$\inf_{\sigma_1 \in \SigmaMin} \costMax(\langle \sigma_1,\sigmaopt_2
  \rangle_v)= \val(v) =\sup_{\sigma_2 \in \SigmaMax} \costMin(\langle
  \sigmaopt_1,\sigma_2 \rangle_v) \,.$$
\end{definition}
If $\sigmaopt_1$ is an optimal strategy for player $\Min$, then he
loses at most $\val(v)$ when playing according to it. On the other
hand, player $\Max$ wins at least $\val(v)$ if he plays according to
an optimal strategy $\sigmaopt_2$ for him.

Examples of classical determined Min-Max cost games can be found
in Section~\ref{sec:appl}.

\section{Results}\label{sec:results}

In this section, we first define a large class of cost games for which
Problem~\ref{pbm 1} can be answered positively
(Theorem~\ref{theo:exists EN}). Then, we study existence of simple Nash
equilibria (Theorems~\ref{theo:mem lin} and~\ref{theo:exists EN
  gen}). To define this interesting class of cost games, we need
the concepts of \emph{\propcost} and \emph{coalition-determined} cost
games.

\begin{definition}\label{def:prefix}
  A multiplayer cost game $\mathcal{G} = (\Pi, V, (V_i)_{i \in \Pi},
  E, (\cost_i)_{i \in \Pi})$ is \emph{\propcost}\ if, for every
  player $i \in \Pi$, every vertex $v \in V$ and history $hv \in
  \hist$, there exists $a \in \IR$ and $b \in \IR^+$ such that, for
  every play $\rho \in \plays$ with $\first(\rho)=v$, we
  have: $$\cost_i(h\rho)= a + b\cdot\cost_i(\rho)\,.$$
\end{definition}

Let us now define the concept of \emph{coalition-determined} cost
games.

\begin{definition}\label{def:pos coal det}
  A multiplayer cost game $\mathcal{G} = (\Pi, V, (V_i)_{i \in \Pi},
  E, (\cost_i)_{i \in \Pi})$ is \emph{(positio\-nally/finite-memory)
    coalition-determined} if, for every player $i \in \Pi$, there
  exists a gain function $\costMax^i: \plays \to
  \IR\cup\{+\infty,-\infty\}$ such that
  \begin{itemize} \item $\cost_i \ge \costMax^i$, and

      \vspace{1ex}
    \item the Min-Max cost game $\mathcal{G}^i=(V,V_i,V \setminus
      V_i,E,\cost_i,\costMax^i)$, where player~$i$ (player $\Min$)
      plays against the coalition $\Pi \setminus \{i\}$ (player
      $\Max$), is determined and has (positional/finite-memory)
      optimal strategies for both players. That is: $\exists\,
      \sigmaopt_i \in \SigmaMin, \ \exists\,\sigmaopt_{-i} \in
      \SigmaMax$ (both positional/finite-memory) such that $\forall v \in
      V$ $$\inf_{\sigma_i \in \SigmaMin} \costMax^i(\langle
      \sigma_i,\sigmaopt_{-i} \rangle_v)= \val^i(v) =\sup_{\sigma_{-i}
        \in \SigmaMax} \cost_i(\langle \sigmaopt_i,\sigma_{-i}
      \rangle_v)\,.$$
    \end{itemize}
\end{definition}
Given $i \in \Pi$, note that $\mathcal{G}^i$ does not depend
on the cost functions $\cost_j$, with $j \not= i$.

\begin{example}\label{excont:cost game}
  Let us consider the two-player cost game $\mathcal{G}$ of
  Example~\ref{ex:cost game}, where player~1 has a quantitative
  reachability objective ($\cost_1$) and player~$2$ has a mean-payoff
  objective ($\cost_2$).  We show that $\mathcal{G}$ is \propcoalpos.

  Let us set $\costMax^1=\cost_1$ and study the Min-Max cost game
  $\mathcal{G}^1=(V,V_1,V_2,\linebreak E,\cost_1,\costMax^1)$, where
  player Min (resp. Max) is player~1 (resp. 2) and wants to minimise
  $\cost_1$ (resp. maximise $\costMax^1$). This game is positionally
  determined \cite{trivedi,filar97}. We define positional strategies
  $\sigmaopt_1$ and $\sigmaopt_{-1}$ for player~1 and player~2,
  respectively, in the following way: $\sigmaopt_1(A)=B$ and
  $\sigmaopt_{-1}(B)=A$. From~$A$, their outcome is $\langle
  (\sigmaopt_1,\sigmaopt_{-1}) \rangle_A = (AB)^\omega$, and
  $\cost_1((AB)^\omega)=\costMax^1((AB)^\omega)=+\infty$. One
  can check that the strategies $\sigmaopt_1$ and $\sigmaopt_{-1}$ are
  optimal in $\mathcal{G}^1$.  Note that the positional strategy
  $\tilde{\sigma}^{\star}_1$ defined by
  $\tilde{\sigma}^{\star}_1(A)=D$ is also optimal (for player~1) in
  $\mathcal{G}^1$. With this strategy, we have that $\langle
  (\tilde{\sigma}^{\star}_1,\sigmaopt_{-1}) \rangle_A = (ADB)^\omega$,
  and $\cost_1((ADB)^\omega)=\costMax^1((ADB)^\omega){=+\infty}$.

  We now examine the Min-Max cost game $\mathcal{G}^2=(V,V_2,V_1,E,
  \cost_2,\costMax^2)$, where $\costMax^2$ is defined as $\cost_2$ but
  with $\liminf$ instead of $\limsup$. In this game, player Min
  (resp. Max) is player~2 (resp. 1) and wants to minimise $\cost_2$
  (resp. maximise $\costMax^2$). This game is also positionally determined
  \cite{trivedi,filar97}. Let $\sigmaopt_2$ and $\sigmaopt_{-2}$ be
  the positional strategies for player~2 and player~1, respectively,
  defined as follows: $\sigmaopt_2(B)=C$ and $\sigmaopt_{-2}(A)=D$.
  From~$A$, their outcome is $\langle (\sigmaopt_2,\sigmaopt_{-2})
  \rangle_A = AD(BC)^\omega$, and $\cost_2(AD(BC)^\omega)=
  \costMax^2(AD(BC)^\omega)=2$.  We claim that $\sigmaopt_2$ and
  $\sigmaopt_{-2}$ are the only positional optimal strategies in
  $\mathcal{G}^2$.
\end{example}

Theorem~\ref{theo:exists EN} positively answers Problem~\ref{pbm 1}
for \propcost, \propcoal\ cost games.

\begin{theorem}\label{theo:exists EN}
  In every initialised multiplayer cost game that is \propcost\ and
  \propcoal, there exists a Nash equilibrium.
\end{theorem}

\begin{proof}
  Let $(\mathcal{G} = (\Pi, V, (V_i)_{i \in \Pi}, E, (\cost_i)_{i \in
    \Pi}),v_0)$ be an initialised multiplayer cost game that is
  \propcost\ and \propcoal. Thanks to the latter property, we know
  that, for every $i \in \Pi$, there exists a gain function
  $\costMax^i$ such that the Min-Max cost game $\mathcal{G}^i=(V,V_i,V
  \setminus V_i,E,\cost_i,\costMax^i)$ is determined and there exist
  optimal strategies $\sigmaopt_i$ and $\sigmaopt_{-i}$ for player $i$
  and the coalition $\Pi \setminus \{i\}$ respectively. In particular,
  for $j \not = i$, we denote by $\sigmaopt_{j,i}$ the strategy of
  player $j$ derived from the strategy $\sigmaopt_{-i}$ of the
  coalition $\Pi \setminus \{i\}$.

  The idea is to define the required Nash equilibrium as follows: each
  player~$i$ plays according to his strategy $\sigmaopt_i$ and
  punishes the first player~$j \not= i$ who deviates from his strategy
  $\sigmaopt_j$, by playing according to $\sigmaopt_{i,j}$ (the
  strategy of player~$i$ derived from $\sigmaopt_{-j}$ in the game
  $\mathcal{G}^j$).

  Formally, we consider the outcome of the optimal strategies
  $(\sigmaopt_i)_{i \in \Pi}$ from $v_0$, and set $\rho:=\langle
  (\sigmaopt_i)_{i \in \Pi} \rangle_{v_0}$. We need to specify a
  punishment function $\pun: \hist \to \Pi \cup \{\bot\}$ that detects
  who is the first player to deviate from the play~$\rho$, i.e. who
  has to be punished. For the initial vertex $v_0$, we define
  $\pun(v_0) = \bot$ (meaning that nobody has deviated from $\rho$)
  and for every history $hv \in \hist$, we let:
  \[ \pun(hv) := \left\{ \begin{array}{ll} 
    \bot & \mbox{ if $\pun(h)= \bot$ and $hv$ is a prefix of
      $\rho$,}\\
    
    i & \mbox{ if $\pun(h)= \bot$, $hv$ is not a prefix of
      $\rho$, and $\last(h) \in V_i$,}\\

    \pun(h) & \mbox{ otherwise ($\pun(h) \not=\bot$).}
  \end{array}\right. \] 
  Then the definition of the Nash equilibrium $(\tau_i)_{\ipi}$ in
  $\mathcal{G}$ is as follows. For all $i \in \Pi$ and $h \in \hist$
  such that $\last(h) \in V_i$, 
  \[ \tau_i(h) :=  \left\{ 
  \begin{array}{ll} 
    \sigmaopt_i(h) & \mbox{ if $\pun(h)=\bot$ or $i$,}\\

    \sigmaopt_{i,\pun(h)}(h) & \mbox{ otherwise.}
  \end{array}\right.\] 
  Clearly the outcome of $(\tau_i)_{\ipi}$ is the play $\rho$
  ($=\langle (\sigmaopt_i)_{i \in \Pi} \rangle_{v_0}$).

  Now we show that the strategy profile $(\tau_i)_{\ipi}$ is a Nash
  equilibrium in $\mathcal{G}$. As a contradiction, let us assume that
  there exists a profitable deviation $\tau_j'$ for some player~$j \in
  \Pi$. We denote by $\rho':= \langle \tau_j', (\tau_i)_{i \in \Pi
    \setminus \{j\}} \rangle_{v_0}$ the outcome where player~$j$ plays
  according to his profitable deviation $\tau_j'$ and the players of
  the coalition $\Pi \setminus \{j\}$ keep their strategies
  $(\tau_i)_{i \in \Pi \setminus \{j\}}$. Since $\tau_j'$ is a
  profitable deviation for player~$j$ w.r.t. $(\tau_i)_{\ipi}$, we
  have that:
  \begin{equation} 
    \cost_j(\rho') <\cost_j(\rho).\label{eq:cost} 
  \end{equation}

  As both plays $\rho$ and $\rho'$ start from vertex $v_0$, there
  exists a history $hv \in \hist$ such that $\rho=h \langle
  (\tau_i)_{\ipi} \rangle_v$ and $\rho'=h \langle \tau_j',(\tau_i)_{i
    \in \Pi \setminus \{j\}} \rangle_{v}$ (remark that $h$ could be
  empty). Among the common prefixes of $\rho$ and $\rho'$, we choose
  the history $hv$ of maximal length. By definition of the strategy
  profile $(\tau_i)_{i \in \Pi}$, we can write in the case of the
  outcome $\rho$ that $\rho=h \langle (\sigmaopt_i)_{\ipi}
  \rangle_v$. Whereas in the case of the outcome $\rho'$, player~$j$
  does not follow his strategy $\sigmaopt_j$ any more from vertex $v$,
  and so, the coalition $\Pi \setminus \{j\}$ punishes him by playing
  according to the strategy $\sigmaopt_{-j}$ after history $hv$, and
  so $\rho'=h \langle \tau_j', \sigmaopt_{-j} \rangle_{v}$ (see
  Figure~\ref{fig:game}).

  \begin{figure}[h!] 
  \centering
    \begin{tikzpicture}[yscale=.3,xscale=.75]
      \everymath{\scriptstyle}
      \draw (0,8) -- (-4,0);
      \draw (0,8) -- (4,0);
      
      \draw[fill=black] (0,8) circle (1.5pt);
      \path (0,8) node[right] (q0) {$v_0$};

      \draw[thick] (0,8) .. controls (-.2,7.1) .. (.3,5.5);
      \path (-.1,7) node[right] (q0) {$h$};

      \draw[fill=black] (.3,5.5) circle (1pt);
      \path (.3,5.5) node[right] (q0) {$v$};

      \draw (.3,5.5) .. controls (1,3.5) .. (1.3,0); 

      \path (2.5,0.1) node[below] (q0) {$\rho\,=\,h \langle
        (\sigmaopt_i)_{\ipi} \rangle_v$};

     \draw[fill=black] (.3,5.5) circle (1.5pt);

      \draw (.3,5.5) .. controls (-1,4) .. (-2,0);

      \path (-.9,.1) node[below] (q0) {$\rho'\,=\,h \langle
        \tau_j', \sigmaopt_{-j} \rangle_{v}$};
      
    \end{tikzpicture}
    \caption{Sketch of the tree representing the unravelling of
        the game $\mathcal{G}$ from $v_0$.}
    \label{fig:game}
  \end{figure}
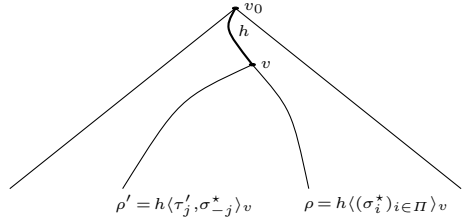

  Since $\sigmaopt_{-j}$ is an optimal strategy for the coalition $\Pi
  \setminus \{j\}$ in the determined Min-Max cost game
  $\mathcal{G}^j$, we have:
  \begin{eqnarray}
    \val^j(v) & = & \inf_{\sigma_j \in \SigmaMin} \costMax^j(\langle
    \sigma_j,\sigmaopt_{-j} \rangle_v) \nonumber \\ & \le &
    \costMax^j(\langle \tau_j', \sigmaopt_{-j} \rangle_{v}) \nonumber
    \\ & \le & \cost_j(\langle \tau_j', \sigmaopt_{-j}
    \rangle_{v})\,. \label{eq:val}
  \end{eqnarray}
  The last inequality comes from the hypothesis $\cost_j \ge
  \costMax^j$ in the game $\mathcal{G}^j$.

  Moreover, the game $\mathcal{G}$ is \propcost, and
    then, when considering the history $hv$, there exist
    $a \in \IR$ and $b \in \IR^+$ such
    that 
    \begin{equation}\label{eq:ineq 1}     
    \cost_j(\rho')= \cost_j(h \langle \tau_j', \sigmaopt_{-j}
    \rangle_{v})= a + b \cdot\cost_j(\langle \tau_j', \sigmaopt_{-j}
    \rangle_{v})\,.
    \end{equation}
    As $b \ge 0$, Equations~\eqref{eq:val} and~\eqref{eq:ineq 1}
    imply:
    \begin{equation}\label{eq:costrho'} 
      \cost_j(\rho') \ge a + b \cdot \val^j(v)\,.  
    \end{equation}

  Since $h$ is also a prefix of $\rho$, we have:
  \begin{equation}\label{eq:ineq 2}
    \cost_j(\rho)= \cost_j(h \langle (\sigmaopt_i)_{\ipi} \rangle_{v})
    = a + b \cdot\cost_j(\langle (\sigmaopt_i)_{\ipi} \rangle_{v})\,.
  \end{equation}
  Furthermore, as $\sigmaopt_{j}$ is an optimal strategy for
  player~$j$ in the Min-Max cost game $\mathcal{G}^j$, it follows
  that:
  \begin{eqnarray}
    \val^j(v) & = & \sup_{\sigma_{-j} \in \SigmaMax} \cost_j(\langle
    \sigmaopt_j,\sigma_{-j} \rangle_v) \nonumber \\ & \ge &
    \cost_j(\langle (\sigmaopt_i)_{\ipi}
    \rangle_{v})\,. \label{eq:val2}
  \end{eqnarray}
  Then, Equations~\eqref{eq:ineq 2} and~\eqref{eq:val2} imply:
  \begin{equation}\label{eq:costrho}
    \cost_j(\rho) \le  a + b \cdot \val^j(v)\,.
  \end{equation}

  Finally, Equations~\eqref{eq:costrho'} and~\eqref{eq:costrho} lead
  to the following inequality:
  $$\cost_j(\rho) \le a + b \cdot \val^j(v) \le \cost_j(\rho')\,,$$
  which contradicts Equation~\eqref{eq:cost}. The strategy profile
  $(\tau_i)_{\ipi}$ is then a Nash equilibrium in the game
  $\mathcal{G}$. 
  \qed
\end{proof}

\begin{remark}
  The proof of Theorem~\ref{theo:exists EN} remains valid for cost
  functions $\cost_i: \plays \to K$, where $\langle K,+,\cdot,0,1,\le
  \rangle$ is an ordered field. This allows for instance to consider
  non-standard real costs and enjoy infinitesimals to model the costs
  of a player.
\end{remark}

\begin{example}\label{excont2:cost game}
  Let us consider the initialised two-player cost game
  $(\mathcal{G},A)$ of Example~\ref{ex:cost game}, where player~1 has
  a quantitative reachability objective ($\cost_1$) and player~$2$ has
  a mean-payoff objective ($\cost_2$). One can show that $\mathcal{G}$
  is \propcost. Since we saw in Example~\ref{excont:cost game} that
  this game is also \propcoalpos, we can apply the construction in the
  proof of Theorem~\ref{theo:exists EN} to get a Nash equilibrium
  in~$\mathcal{G}$.  The construction from this proof may result in
  two different Nash equilibria, depending on the selection of the
  strategies $\sigmaopt_1$/$\tilde{\sigma}^\star_1$, $\sigmaopt_{-1}$,
  $\sigmaopt_2$ and $\sigmaopt_{-2}$ as defined in
  Example~\ref{excont:cost game}.

  The first Nash equilibrium $(\tau_1,\tau_2)$ with outcome
  $\rho=\langle \sigmaopt_1,\sigmaopt_2 \rangle_A=A(BC)^\omega$ is
  given, for any history~$h$, by:
  \[ \tau_1(hA) :=  \left\{ 
  \begin{array}{ll} 
    B & \mbox{ if $\pun(hA)=\{\bot,1\}$}\\
    
    D & \mbox{ otherwise}
  \end{array}\right. \quad ; \quad 
  \tau_2(hB) :=  \left\{ 
  \begin{array}{ll} 
    C & \mbox{ if $\pun(hB)=\{\bot,2\}$}\\
    
    A & \mbox{ otherwise}
  \end{array}\right. \]
  where the punishment function $\pun$ is defined as in the proof of
  Theorem~\ref{theo:exists EN} and depends on the play $\rho$. The
  cost for this finite-memory Nash equilibrium is
  $\cost_1(\rho)=2=\cost_2(\rho)$.

 The strategy $\tilde{\tau}_1$ of the second Nash equilibrium
 $(\tilde{\tau}_1,\tau_2)$ with outcome $\tilde{\rho}=\langle
 \tilde{\sigma}^\star_1,\sigmaopt_2 \rangle_A=AD(BC)^\omega$ is given
 by $\tilde{\tau}_1(hA) := D$ for all history~$h$. The cost for this
 finite-memory Nash equilibrium is $\cost_1(\tilde{\rho})=6$ and
 $\cost_2(\tilde{\rho})=2$, respectively.

 Note that there is no positional Nash equilibrium with outcome $\rho$
 (resp. $\tilde{\rho}$).
\end{example}

The two following theorems provide results about the complexity of the
Nash equilibrium defined in the latter proof. Applications of these
theorems to specific classes of cost games are provided in
Section~\ref{sec:appl}.

\begin{theorem}\label{theo:mem lin}
  In every initialised multiplayer cost game that is \propcost\ and
  \propcoalposemph, there exists a Nash equilibrium with memory (at
  most) $|V|+|\Pi|$.
\end{theorem}

\begin{theorem}\label{theo:exists EN gen}
  In every initialised multiplayer cost game that is \propcost\
  and \emph{finite-memory} coalition-determined, there exists a Nash
  equilibrium with finite memory.
\end{theorem}

The proofs of these two theorems rely on the construction of the Nash
equilibrium provided in the proof of Theorem~\ref{theo:exists EN}.

\section{Applications}\label{sec:appl}

In this section, we exhibit several classes of \emph{classical
  objectives} that can be encoded in our general setting. The list we
propose is far from being exhaustive.

\subsection{Qualitative Objectives}
Multiplayer games with qualitative (win/lose) objectives can naturally be
encoded via multiplayer cost games; for instance via cost functions
$\cost_i:\plays \to \{1,+\infty\}$, where $1$ (resp. $+\infty$) means
that the play is won (resp. lost) by player~$i$. Let us now consider
the subclass of qualitative games with prefix-independent\footnote{An
  objective $\Omega \subseteq V^\omega$ is prefix-independent if only
  if for every play $\rho = \rho_0 \rho_1 \ldots \in V^\omega$, we
  have that $\rho \in \Omega$ iff for every $n \in \IN$, $ \rho_n
  \rho_{n+1} \ldots \in \Omega$.} Borel objectives. Given
such a game $\mathcal{G}$, we have that $\mathcal{G}$ is
coalition-determined, as a consequence of the Borel determinacy
theorem~\cite{Ma75}.  Moreover the prefix-independence hypothesis
obviously guarantees that $\mathcal{G}$ is also \propcost\ (by taking
$a=0$ and $b=1$). By applying Theorem~\ref{theo:exists EN}, we obtain
the existence of a Nash equilibrium for qualitative games with
prefix-independent Borel objectives. Let us notice that this result
is already present in~\cite{GU08}.

When considering more specific subclasses of qualitative games
enjoying a positional determinacy result, such as parity
games~\cite{GTW02}, we can apply Theorem~\ref{theo:mem lin} and ensure
 existence of a Nash equilibrium whose memory is (at most)
linear.

\subsection{Classical Quantitative Objectives}

We here give four well-known kinds of Min-Max cost games and see later
that they are determined. For each sort of game, the cost and gain
functions are defined from a price function (and a reward function in
the last case), which labels the edges of the game graph with prices
(and rewards).

\begin{definition}[\cite{trivedi}]\label{def:part cost games}
  Given a game graph $G=(V,\VMin,\VMax,E)$, a price function $\price:
  E \to \IR$ that assigns a price to each edge, a
  diverging\footnote{For all plays $\rho=\rho_0\rho_1\ldots$ in $G$,
    it holds that $\lim_{n \to \infty} |\sum_{i=1}^n
    \reward(\rho_{i-1},\rho_i)| = +\infty$. This is equivalent to
    requiring that every cycle has a positive sum of rewards.}  reward
  function $\reward: E \to \IR$ that assigns a reward to each edge,
  and a play $\rho=\rho_0\rho_1\ldots$ in $G$, we define the following
  Min-Max cost games:
  \begin{enumerate}
  \item a \emph{reachability-price game} is a Min-Max cost game
    $\mathcal{G} = (G,\RPMin,\RPMax)$ together with a given goal set $\FMin
    \subseteq V$, where
    \[  \RPMin(\rho)=\RPMax(\rho) = \left\{
    \begin{array}{ll}
        \price(\rho[0,n]) & \mbox{ if
        $n$ is the \emph{least} index s.t. $\rho_n \in
        \FMin$,}\\ + \infty & \mbox{ otherwise,}
    \end{array}\right. \]

    with $\price(\rho[0,n])= \sum_{i=1}^n \price(\rho_{i-1},\rho_i)$;
    \smallskip
  \item a \emph{discounted-price game} is a Min-Max cost game
    $\mathcal{G} = (G,\DPMin(\lambda),\DPMax(\lambda))$ together with
    a given discount factor $\lambda \in\, \left]0,1\right[$, where
    $$\DPMin(\lambda)(\rho)=\DPMax(\lambda)(\rho)=
        (1-\lambda) \cdot \sum_{i=1}^{+\infty} \lambda^{i-1}
        \price(\rho_{i-1},\rho_i)\,;$$

  \item an \emph{average-price game}\footnote{When the cost function
    of a player is $\APMin$, we  say that he has a
    mean-payoff objective.} is a Min-Max cost game $\mathcal{G} =
    (G,\APMin,\APMax)$, where
    $$\APMin(\rho)=\displaystyle \limsup_{n \to +\infty}
    \frac{\price(\rho[0,n])}{n}  \quad\text{and}\quad
    \APMax(\rho)= \displaystyle \liminf_{n \to +\infty}
    \frac{\price(\rho[0,n])}{n}\,;$$

  \item a \emph{price-per-reward-average game} is a Min-Max cost game
    $\mathcal{G} = (G,\PRAvgMin,\linebreak\PRAvgMax)$, where
    $$\PRAvgMin(\rho)=\displaystyle \limsup_{n \to +\infty}
    \frac{\price(\rho[0,n])}{\reward(\rho[0,n])} \quad\text{and}\quad
    \PRAvgMax(\rho)=\displaystyle \liminf_{n \to +\infty}
    \frac{\price(\rho[0,n])}{\reward(\rho[0,n])}\,,$$ with
    $\reward(\rho[0,n])=\sum_{i=1}^n \reward(\rho_{i-1},\rho_i)$.
  \end{enumerate}
\end{definition}
An average-price game is then a particular case of a
price-per-reward-average game. Let us remark that, in
Example~\ref{ex:cost game}, the cost function $\cost_1$
(resp. $\cost_2$) corresponds to $\RPMin$ with $\FMin=\{C\}$
(resp. $\APMin$). The game $\mathcal{G}^1$ (resp. $\mathcal{G}^2$) of
Example~\ref{excont:cost game} is a reachability-price
(resp. average-price) game.

The following theorem is a well-known result about the particular cost
games described in Definition~\ref{def:part cost games}.

\begin{theorem}[\cite{trivedi,filar97}]\label{theo:determined}
  Reachability-price games, discounted-price games,
  average-price games, and price-per-reward games are determined and
  have positional optimal strategies.
\end{theorem}
This result implies that a multiplayer cost game where each cost
function is $\RPMin$, $\DPMin$, $\APMin$ or $\PRAvgMin$ is
\propcoalpos. Moreover, one can show that such a game is \propcost.
Theorem~\ref{theo:exists EN PC} then follows from
Theorem~\ref{theo:mem lin}.

\begin{theorem}\label{theo:exists EN PC}
 In every initialised multiplayer cost game $\mathcal{G} = (\Pi, V,
 (V_i)_{i \in \Pi}, E,\linebreak (\cost_i)_{i \in \Pi})$ where the
 cost function $\cost_i$ belongs to $\{\RPMin,\DPMin,\APMin,\linebreak
 \PRAvgMin\}$ for every player $i \in \Pi$, there exists a Nash
 equilibrium with memory (at most) $|V|+|\Pi|$.
\end{theorem}
Note that the existence of finite-memory Nash equilibria in cost games
with quantitative reachability objectives has already been established
in~\cite{BBD10,BBD12}. Even if not explicitly stated in the previous
papers, one can deduce from the proof
of~\cite[Lemma~16]{BBD12} that the provided Nash equilibrium has a
memory (at least) exponential in the size of the cost game. Thus,
Theorem~\ref{theo:exists EN PC} significantly improves the complexity
of the strategies constructed in the case of cost games with
quantitative reachability objectives.

\subsection{Combining Qualitative and Quantitative Objectives}
Multiplayer cost games allow to encode games combining both
qualitative and quantitative objectives, such as \emph{mean-payoff
  parity games}~\cite{CHJ05}. In our framework, where each player aims
at minimising his cost, the mean-payoff parity objective could be
encoded as follows: $\cost_i(\rho) = \APMin(\rho)$ if the parity
condition is satisfied, $+ \infty$ otherwise.

 The determinacy of mean-payoff parity games, together with the
 existence of optimal strategies (that could require infinite memory)
 have been proved in~\cite{CHJ05}. This result implies that
 multiplayer cost games with mean-payoff parity objectives are
 coalition-determined. Moreover, one can prove that such a game is
 also \propcost\ (by taking $a=0$ and $b=1$). By applying
 Theorem~\ref{theo:exists EN}, we obtain the existence of a Nash
 equilibrium for multiplayer cost games with mean-payoff parity
 objectives. As far as we know, this is the first result about the
 existence of a Nash equilibrium in cost games with mean-payoff parity
 games.

\begin{remark}
Let us emphasise that Theorem~\ref{theo:exists EN} applies to cost
games where the players have different kinds of cost functions (as in
Example~\ref{ex:cost game}). In particular, one player could have a
qualitative B\"uchi objective, a second player a discounted-price
objective, a third player a mean-payoff parity objective,\ldots
\end{remark}

\bibliographystyle{abbrv}
\bibliography{main}

\begin{thebibliography}{10}

\bibitem{AKW08}
R.~Alur, A.~Kanade, and G.~Weiss.
\newblock Ranking automata and games for prioritized requirements.
\newblock In {\em CAV}, volume 5123 of {\em LNCS}, pages 240--253. Springer,
  2008.

\bibitem{BCHJ09}
R.~Bloem, K.~Chatterjee, T.~Henzinger, and B.~Jobstmann.
\newblock Better quality in synthesis through quantitative objectives.
\newblock In {\em CAV}, volume 5643 of {\em LNCS}, pages 140--156. Springer,
  2009.

\bibitem{BG09}
E.~Boros and V.~Gurvich.
\newblock Why chess and back gammon can be solved in pure positional uniformly
  optimal strategies.
\newblock Rutcor Research Report 21-2009, Rutgers University, 2009.

\bibitem{BBM10}
P.~Bouyer, R.~Brenguier, and N.~Markey.
\newblock Nash equilibria for reachability objectives in multi-player timed
  games.
\newblock In {\em Concurrency Theory, CONCUR}, volume 6269 of {\em LNCS}, pages
  192--206. Springer, 2010.

\bibitem{BBMU11}
P.~Bouyer, R.~Brenguier, N.~Markey, and M.~Ummels.
\newblock {N}ash equilibria in concurrent games with {B}{\"u}chi objectives.
\newblock In {\em FSTTCS}, volume~13 of {\em LIPIcs}, pages 375--386. Schloss
  Dagstuhl - Leibniz-Zentrum fuer Informatik, 2011.

\bibitem{BFLMS08}
P.~Bouyer, U.~Fahrenberg, K.~G. Larsen, N.~Markey, and J.~Srba.
\newblock Infinite runs in weighted timed automata with energy constraints.
\newblock In {\em FORMATS}, volume 5215 of {\em LNCS}, pages 33--47. Springer,
  2008.

\bibitem{BBD10}
T.~Brihaye, V.~Bruy{\`e}re, and J.~De~Pril.
\newblock Equilibria in quantitative reachability games.
\newblock In {\em CSR}, volume 6072 of {\em LNCS}, pages 72--83. Springer,
  2010.

\bibitem{BBD12}
T.~Brihaye, V.~Bruy{\`e}re, and J.~De~Pril.
\newblock On equilibria in quantitative games with reachability/safety
  objectives.
\newblock {\em CoRR}, abs/1205.4889, 2012.

\bibitem{BBDG12}
T.~Brihaye, V.~Bruy{\`e}re, J.~De~Pril, and H.~Gimbert.
\newblock Subgame perfection for equilibria in quantitative reachability games.
\newblock In {\em FoSSaCS}, volume 7213 of {\em LNCS}, pages 286--300.
  Springer, 2012.

\bibitem{CH06}
K.~Chatterjee and T.~A. Henzinger.
\newblock Finitary winning in omega-regular games.
\newblock In {\em TACAS}, volume 3920 of {\em LNCS}, pages 257--271. Springer,
  2006.

\bibitem{CHJ05}
K.~Chatterjee, T.~A. Henzinger, and M.~Jurdzinski.
\newblock Mean-payoff parity games.
\newblock In {\em LICS}, pages 178--187. IEEE Computer Society, 2005.

\bibitem{CGP00}
E.~Clarke, O.~Grumberg, and D.~Peled.
\newblock {\em Model Checking}.
\newblock MIT Press, Cambridge, MA, 2000.

\bibitem{FZ12}
N.~Fijalkow and M.~Zimmermann.
\newblock Cost-parity and cost-streett games.
\newblock {\em CoRR}, abs/1207.0663, 2012.

\bibitem{filar97}
J.~Filar and K.~Vrieze.
\newblock {\em Competitive Markov decision processes}.
\newblock Springer Verlag, 1997.

\bibitem{GTW02}
E.~Gr\"adel, W.~Thomas, and T.~Wilke, editors.
\newblock {\em Automata, logics, and infinite games}, volume 2500 of {\em
  LNCS}. Springer, 2002.

\bibitem{GU08}
E.~Gr\"adel and M.~Ummels.
\newblock Solution concepts and algorithms for infinite multiplayer games.
\newblock In {\em New Perspectives on Games and Interaction}, volume~4 of {\em
  Texts in Logic and Games}, pages 151--178. Amsterdam University Press, 2008.

\bibitem{HU}
J.~E. Hopcroft and J.~D. Ullman.
\newblock {\em Introduction to automata theory, languages, and computation}.
\newblock Addison-Wesley Publishing Co., Reading, Mass., 1979.
\newblock Addison-Wesley Series in Computer Science.

\bibitem{HTW08}
F.~Horn, W.~Thomas, and N.~Wallmeier.
\newblock Optimal strategy synthesis in request-response games.
\newblock In {\em ATVA}, volume 5311 of {\em LNCS}, pages 361--373. Springer,
  2008.

\bibitem{KLST12}
M.~Klimo{\v{s}}, K.~Larsen, F.~{\v{S}}tefa{\v{n}}{\'a}k, and J.~Thaarup.
\newblock {Nash Equilibria in Concurrent Priced Games}.
\newblock In {\em LATA}, volume 7183 of {\em LNCS}, pages 363--376. Springer,
  2012.

\bibitem{Ma75}
D.~A. Martin.
\newblock Borel determinacy.
\newblock {\em Ann. of Math. (2)}, 102(2):363--371, 1975.

\bibitem{Na50}
J.~Nash.
\newblock {Equilibrium points in n-person games}.
\newblock {\em Proceedings of the National Academy of Sciences of the United
  States of America}, 36(1):48--49, 1950.

\bibitem{OR94}
M.~Osborne and A.~Rubinstein.
\newblock {\em A course in game theory}.
\newblock MIT Press, Cambridge, MA, 1994.

\bibitem{PS09}
S.~Paul, S.~Simon, R.~Kannan, and K.~Kumar.
\newblock Nash equilibrium in generalised muller games.
\newblock In {\em FSTTCS}, volume~4 of {\em LIPIcs}, pages 335--346. Schloss
  Dagstuhl - Leibniz-Zentrum fuer Informatik, 2009.

\bibitem{T95}
W.~Thomas.
\newblock On the synthesis of strategies in infinite games.
\newblock In {\em STACS}, volume 900 of {\em LNCS}, pages 1--13. Springer,
  Berlin, 1995.

\bibitem{Th08}
W.~Thomas.
\newblock Church's problem and a tour through automata theory.
\newblock In {\em Pillars of Computer Science}, volume 4800 of {\em LNCS},
  pages 635--655. Springer, 2008.

\bibitem{TT98}
F.~Thuijsman and T.~E.~S. Raghavan.
\newblock Perfect information stochastic games and related classes.
\newblock {\em International Journal of Game Theory}, 26(3):403--408, 1998.

\bibitem{trivedi}
A.~Trivedi.
\newblock {\em {Competative optimisation on timed automata}}.
\newblock PhD thesis, University of Warwick, 2009.

\bibitem{Z09}
M.~Zimmermann.
\newblock Time-optimal winning strategies for poset games.
\newblock In {\em CIAA}, volume 5642 of {\em LNCS}, pages 217--226. Springer,
  2009.

\end{thebibliography}

\newpage
\appendix

\section*{Technical Appendix}
\pagenumbering{roman} 

\section{Example of a cost game which is not \propcost}
\begin{example}
Multiplayer cost games allow to encode energy games.  Let
$\mathcal{G}$ be a cost game defined by means of a price function
$\price: E \to \IR$, that assigns a price to each edge.  In our
framework, where each player aims at minimising his cost, an energy
objective~\cite{BFLMS08} (with threshold $T \in \IR$) could be encoded
as follows:
$$
\cost_i(\rho) =
\begin{cases}
\sup_{n \ge 0} \price(\rho[0,n]) &  \text{ if } \sup_{n \ge 0} 
\price(\rho[0,n]) \le T\\ 
{} \ + \infty & \text{ otherwise,}
\end{cases}
$$
with $\price(\rho[0,n])= \sum_{i=1}^n \price(\rho_{i-1},\rho_i)$.

  \begin{figure}[h!]
    \centering
    \begin{tikzpicture}[xscale=1,yscale=1,>=stealth',shorten >=1pt]
      \everymath{\scriptstyle}
 
      \path (0,0) node[draw,circle,inner sep=2pt] (qA) {$A$};
      \path (2.5,0) node[draw,circle,inner sep=2pt] (qB) {$B$};
     
      \path (qA) edge[loop above] node[pos=.5,above] {{+1}} (qA);
      \path (qB) edge[loop above] node[pos=.5,above] {{-1}} (qB);
      \path (qA) edge[->,bend left] node[pos=.5,above] {{+1}} (qB);
      \path (qB) edge[->,bend left] node[pos=.5,below] {{-1}} (qA);
    \end{tikzpicture}

    \caption{A cost game which is not \propcost}
    \label{fig-notpropcost} 
  \end{figure}

Let us consider the one-player cost game with an energy objective
(with threshold $T=2$) depicted in Figure~\ref{fig-notpropcost}. We
show that this game is not \propcost. For this, we exhibit a
history $hv \in \hist$ such that for all $a, \, b \in \IR$ there
exists a play $\rho \in \plays$ with $\first(\rho)=v$, such that
$\cost_1(h\rho) \ne a + b\cdot\cost_1(\rho)$. We in fact give a
play $\rho$ independent of $a$ and $b$.  Let $hv$ be the history
$AAABA$ and $\rho$ be the play $(AB)^\omega$. We have that
$\cost_1(\rho)=1$ and $\cost_1(h\rho)=\cost_1(AA(AB)^\omega)=+\infty$,
since $\sup_{n \ge 0} \price((h\rho)[0,n])=3$, which is above the
threshold $T=2$. It is thus impossible to find $a, \, b \in \IR$ such
that:
$$
+\infty =\cost_1(h\rho) = a + b\cdot\cost_1(\rho) = a + b.
$$
\end{example}

%\noindent \hrulefill

\section{Remark about secure and subgame perfect equilibria}

\begin{remark}
It would be tempting to try to prove the existence of \emph{subgame
  perfect equilibria} or \emph{secure equilibria}\footnote{The
  definitions of subgame perfect and secure equilibria in this context
  can be found in \cite{BBDG12}.} in multiplayer cost games with
techniques similar to the proof of Theorem~\ref{theo:exists
  EN}. However, our definition of the Nash equilibrium in the proof of
Theorem~\ref{theo:exists EN} is (in general) neither a subgame perfect
equilibrium, nor a secure equilibrium. To see this, let us consider
the following two cost games $\mathcal{G}$ and $\mathcal{H}$, whose
graphs are depicted on Figure~\ref{fig:G1} and~\ref{fig:G2}
respectively. Both games are initialised in vertex $A$.

  \begin{figure}
    \null\hfill
    \begin{minipage}[b]{0.32\linewidth}
      \begin{center}
        \begin{tikzpicture}[xscale=1,yscale=1]
          \everymath{\scriptstyle}
          
          \path (0,2) node[draw,rectangle,inner sep=3.5pt] (q0) {$A$};
          \path (-1,1) node[draw,circle,inner sep=2pt] (q00) {$B$};
          \path (1,1) node[draw,circle,inner sep=2pt] (q01) {$C$};
          \path (-1,0) node[draw,double,circle,inner
            sep=2pt,fill=black!20!white] (q000) {$D$};
          \path (.5,0) node[draw,double,circle,inner sep=2pt,
            fill=black!20!white] (q010) {$E$};
          \path (1.5,0) node[draw,circle,inner sep=2pt] (q011) {$F$};
          
          \draw[arrows=->] (q0) -- (q00) ;
          \draw[arrows=->] (q0) -- (q01) ;
          
          \draw[arrows=->] (q00) -- (q000) ; 
          \draw[arrows=->] (q01) -- (q010) ; 
          \draw[arrows=->] (q01) -- (q011) ;
        \end{tikzpicture}
        
        \caption{Game $\mathcal{G}$.}
        \label{fig:G1}
      \end{center}
    \end{minipage}
    \hfill
    \begin{minipage}[b]{0.32\linewidth}
      \begin{center}
        \begin{tikzpicture}[xscale=1,yscale=1]
          \everymath{\scriptstyle}
          
%  \path (0,0) node[draw,rectangle,inner sep=3.5pt] (q0) {$A$};
%  \path (1,1) node[draw,rectangle,inner sep=3.5pt
%    ,fill=black!20!white] (q00) {$B$};
%  \path (1,-1) node[draw,rectangle,inner sep=3.5pt] (q01) {$C$};
%  \path (2,0) node[draw,circle,inner sep=2pt] (q000) {$D$};
%  \path (3,1) node[draw,circle,inner sep=2pt] (q010) {$E$};
%  \path (3,-1) node[draw,double,circle,inner sep=2pt,
%    fill=black!20!white] (q011) {$F$};

%  \draw[arrows=->] (q0) -- (q00) ;
%  \draw[arrows=->] (q0) -- (q01) ;

%  \draw[arrows=->] (q00) -- (q000) ; 
%  \draw[arrows=->] (q01) -- (q000) ; 
%  \draw[arrows=->] (q000) -- (q010) ;
%  \draw[arrows=->] (q000) -- (q011) ;

          \path (0,1) node[draw,circle,inner sep=2pt,
            fill=black!20!white] (q) {$A$};
          \path (0,0) node[draw,circle,inner sep=2pt] (q0) {$B$};
          \path (-1,-1) node[draw,rectangle,inner sep=3.5pt,
            fill=black!20!white,double]
          (q00) {$C$};
          \path (1,-1) node[draw,rectangle,inner sep=3.5pt]
          (q01) {$D$};

          \draw[arrows=->] (q) -- (q0) ;
          \draw[arrows=->] (q0) -- (q00) ;
          \draw[arrows=->] (q0) -- (q01) ;
          
        \end{tikzpicture}
        \caption{Game $\mathcal{H}$.}
        \label{fig:G2}
      \end{center}
    \end{minipage}    
    \hfill\null
  \end{figure}

The game $\mathcal{G}$ is a two-player cost game where the vertices of
player~1 (resp.~2) are represented by circles (resp. squares), that
is, $V_1=\{B,C,D,E,F\}$ and $V_2=\{A\}$. The cost functions of both
players are $\RPMin$, with\footnote{In both figures, shaded
  (resp. doubly circled) vertices represent the goal set $\FMin_1$
  (resp. $\FMin_2$).} $\FMin_1= \FMin_2= \{D,E\}$ and the price
function $\price:E \to \IR$ defined by $\price(e)=1$ for any edge $e
\in E$ (same price function for the two players). It means that both
players have reachability objectives and want to reach vertex $D$ or
$E$ within the least number of edges.

%$V_1=\{D\}$ and $V_2=\{A,B,C,E,F\}$.  
%$\FMin_1=\{B,F\}$, $\FMin_2=\{F\}$ in $\mathcal{H}$. Moreover, the
%price functions of the two games assign 1 to every edge.

Let us study the two Min-Max cost games $\mathcal{G}^1$ and
$\mathcal{G}^2$. In the game $\mathcal{G}^1$, let $\sigmaopt_1$ be
defined as $\sigmaopt_1(C)=E$ and $\sigmaopt_{-1}$ be defined as
$\sigmaopt_{-1}(A)=C$. Then, $\sigmaopt_1$ and $\sigmaopt_{-1}$ are
positional optimal strategies for player Min (player~1) and player Max
(player~2) respectively. In the game $\mathcal{G}^2$, we define
$\sigmaopt_2$ and $\sigmaopt_{-2}$ as $\sigmaopt_2(A)=B$ and
$\sigmaopt_{-2}(C)=F$. These two strategies of $\mathcal{G}^2$ are
positional optimal strategies for player Min (player~2) and player Max
(player~1) respectively.

If we define a Nash equilibrium $(\tau_1,\tau_2)$ in $\mathcal{G}$
exactly as in the proof of Theorem~\ref{theo:exists EN}, depending
on these strategies $\sigmaopt_1$, $\sigmaopt_{-1}$, $\sigmaopt_2$ and
$\sigmaopt_{-2}$, then $(\tau_1,\tau_2)$ is not a subgame perfect
equilibrium in $\mathcal{G}$. Indeed, $(\tau_1|_A,\tau_2|_A)$ is not a
Nash equilibrium in the subgame $\mathcal{G}|_A$ with history $AC$:
player~1 punishes player~2 by choosing the edge $(C,F)$ (according to
$\sigmaopt_{-2}$) whereas player~1 could pay a smaller cost by
choosing the edge $(C,E)$.

Furthermore, this Nash equilibrium also gives a counter-example of
subgame perfect equilibrium for other classical punishments (see
\cite{OR94}, e.g., punish the last player who has deviated and only
for a finite number of steps).

Let us now consider the two-player cost game $\mathcal{H}$ where
$V_1=\{A,B\}$ and $V_2=\{C,D\}$ (see Figure~\ref{fig:G2}). The price
function and the cost functions of the two players are the same as in
the game $\mathcal{G}$, except that here $\FMin_1=\{A,C\}$ and
$\FMin_2=\{C\}$. Note that player~$2$ does not really play in
$\mathcal{H}$, only player~$1$ has a choice to make: he can choose the
edge $(B,C)$ or the edge $(B,D)$.

As before, we study the two Min-Max cost games $\mathcal{H}^1$ and
$\mathcal{H}^2$. Let $\sigmaopt_1$ be a positional strategy of
player~1 in $\mathcal{H}^1$ such that $\sigmaopt_1(B)=C$, and
$\sigmaopt_{-2}$ be a positional strategy of player~1 in
$\mathcal{H}^2$ such that $\sigmaopt_{-2}(B)=D$. These strategies are
optimal in the two respective games.  Then, we define a Nash
equilibrium in $\mathcal{H}$ in the same way as in the proof of
Theorem~\ref{theo:exists EN}, depending on $\sigmaopt_1$ and
$\sigmaopt_{-2}$. Actually, this is not a secure equilibrium in
$\mathcal{H}$ because player~1 can strictly increase player~2's cost
while keeping his own cost, by choosing the edge $(B,D)$ instead of
following $\sigmaopt_1$ ($\sigmaopt_1$ suggests to choose the edge
$(B,C)$).
\end{remark}

%\noindent \hrulefill

\section{Proof of Theorem~\ref{theo:mem lin}}

  Theorem~\ref{theo:mem lin} states that in every initialised
  multiplayer cost game that is \propcost\ and \propcoalposemph, there
  exists a Nash equilibrium with memory (at most) $|V|+|\Pi|$.

\begin{proof}
 Let $(\mathcal{G} = (\Pi, V, (V_i)_{i \in \Pi}, E, (\cost_i)_{i \in
   \Pi}),v_0)$ be an initialised multiplayer cost game that is
 \propcost\ and \propcoalposemph. For this proof, we keep the
 notations introduced in the proof of Theorem~\ref{theo:exists EN}. In
 particular, we consider the Nash equilibrium $(\tau_i)_{\ipi}$ as
 defined in the latter proof, whose outcome is $\rho:=\langle
 (\sigmaopt_i)_{i \in \Pi} \rangle_{v_0}$. We recall that for all $i
 \in \Pi$, the strategy $\tau_i$ depends on the strategies
 $\sigmaopt_i$ (optimal strategy in $\mathcal{G}^i$) and
 $\sigmaopt_{i,j}$ (derived from the optimal strategy $\sigmaopt_{-j}$
 in $\mathcal{G}^j$) for $j \in \Pi \setminus \{i\}$. As the game
 $\mathcal{G}$ is now \propcoalpos\ by hypothesis, these strategies
 are assumed to be positional. This proof consists in showing that
 $(\tau_i)_{\ipi}$ is a strategy profile with memory (at most)
 $|V|+|\Pi|$.

 For this purpose, we define a finite strategy automaton for each
 player that remembers the play $\rho$ and who has to be punished. As
 the play $\rho$ is the outcome of the positional strategy profile
 $(\sigmaopt_i)_{i \in \Pi}$, we can write $\rho:=v_0\ldots
 v_{k-1}(v_k\ldots v_n)^{\omega}$ where $0 \le k \le n \le |V|$, $v_l
 \in V$ for all $0 \le l \le n$ and these vertices are all
 different. For any $i \in \Pi$, let $\mathcal{A}_i=
 (M,m_0,V,\delta,\nu)$ be the strategy automaton of player~$i$, where:
  \begin{itemize}
  \item $M = \{v_0v_0,v_0v_1,\ldots,v_{n-1}v_n,v_nv_k\} \cup \Pi
    \setminus \{i\}$.

    As we want to be sure that the play $\rho$ is followed by all
    players, we need to memorise which movement (edge) has to be
    chosen at each step of $\rho$. This is the role of
    $\{v_0v_0,v_0v_1,\ldots,v_{n-1}v_n,v_nv_k\}$. But in case a player
    deviates from $\rho$, we only have to remember this player during
    the rest of the play (no matter if another player later deviates
    from $\rho$). This is the role of $\Pi \setminus \{i\}$.
  \item $m_0=v_0v_0$ (this memory state means that the play has not
    begun yet).
  \item $\delta: M \times V \to M$ is defined in this way: given $m
    \in M$ and $v \in V$,
    \[ \delta(m,v) := 
    \left\{ \begin{array}{ll}
      
      j & \mbox{ if $m=j \in \Pi$ or}\\

      & \mbox{ \quad ($m=u_1u_2$, with $u_1,u_2 \in V$,
        $v\not=u_2$ and $u_1 \in V_j$),}\\
      
      v_lv_{l+1} & \mbox{ if $m=uv_l$ for a certain $l \in
        \{0,\ldots,n-1\}$, $u\in V$,}\\
      
      & \mbox{ \quad  and $v=v_l$,} \\
      
      v_nv_{k} & \mbox{ otherwise ($m=uv_n$ and $v=v_n$).}

    \end{array}\right.\]
    Intuitively, $m$ represents either a player to punish, or the edge
    that should, if following $\rho$, have been chosen at the last
    step of the current stage of the play, and $v$ is the real last
    vertex of the current stage of the play.

    Notice that in this definition of $\delta$, $j$ is different from
    $i$ because if player~$i$ follows the strategy computed by this
    strategy automaton, one can be convinced that he does not deviate
    from the play $\rho$.

    \item $\nu: M \times V_i \to V$ is defined in this way: given $m
      \in M$ and $v \in V_i$, 
      \[ \nu(m,v) := \left\{ \begin{array}{ll}
      
      \sigmaopt_i(v) & \mbox{ if $m=u_1u_2$ with $u_1,u_2 \in V$ and
      $v=u_2$,}\\
      
      \sigmaopt_{i,j}(v) & \mbox{ if $m=j \in \Pi$ or}\\
      
      & \mbox{\quad ($m=u_1u_2$, with $u_1,u_2 \in V$, $v\not=u_2$ and
        $u_1 \in V_j$).}

    \end{array}\right.\]
    The idea is to play according to $\sigmaopt_i$ if everybody
    follows the play $\rho$, and switch to $\sigmaopt_{i,j}$ if
    player~$j$ is the first player who has deviated from $\rho$.
  \end{itemize}
  Obviously, the strategy $\sigma_{\mathcal{A}_i}$ computed by
  the strategy automaton ${\mathcal{A}_i}$ exactly corresponds to the
  strategy $\tau_i$ of the Nash equilibrium. And so, we can conclude
  that each strategy $\tau_i$ requires a memory of size at most $|M|
  \le |\Pi| + |V|$. \qed
\end{proof}

%\noindent \hrulefill

\section{Example~\ref{ex:cost game} continued}

\begin{example}\label{exend:cost game}
Thanks to the proof of Theorem~\ref{theo:mem lin}, we can construct a
finite strategy automaton $\mathcal{A}_1$ that computes the
strategy~$\tau_1$ of player~1 given in Example~\ref{excont2:cost
  game}. The set $M$ of memory states is $M=\{AA,AB,BC,CB\} \cup
\{2\}$ since $\rho=A(BC)^\omega$, and the initial state is
$m_0=AA$. The memory update function $\delta: M \times V \to M$ and
the transition choice function $\nu: M \times V_1 \to V$ are depicted
in Figure~\ref{fig:strat aut}: a label $v/v'$ on an edge $(m_1,m_2)$
means that $\delta(m_1,v)=m_2$, and $\nu(m_1,v)=v'$ if $v \in V_1$.
If $v \notin V_1$, we indicate that $\nu$ does not return any advice
by a `$-$', and label the edge with $v/-$.

 \begin{figure}[h!]
    \centering
    \begin{tikzpicture}[xscale=.9,yscale=.7,>=stealth',shorten >=1pt]
      \everymath{\scriptstyle}

      \path (0,0) node[draw,circle,inner sep=2pt] (qA) {$AA$};
      \path (2.5,0) node[draw,circle,inner sep=2pt] (qB) {$AB$};
      \path (5,0) node[draw,circle,inner sep=2pt] (qC) {$BC$};
      \path (7.5,1) node[draw,circle,inner sep=2pt] (qD) {$CB$};
      \path (7.5,-1) node[draw,circle,inner sep=4pt] (qE) {$2$};
      
      \path (qA) edge[->] node[pos=.5,above] {{$A/B$}} (qB);
      \path (qB) edge[->] node[pos=.5,above] {{$B/-$}} (qC);

      \path (qC) edge[->,bend right] node[pos=.6,below] {{$C/B$}} (qD);
      \path (qD) edge[->,bend right] node[pos=.5,above] {{$B/-$}} (qC);

      \path (qC) edge[->,bend right] node[pos=.5,below] {{$A/D$}} (qE);
      \path (qE) edge[->,loop right] node[pos=.6,below] {{$A/D$}} (qE);

    \end{tikzpicture} 
    \caption{The finite strategy automaton $\mathcal{A}_1$.}  
    \label{fig:strat aut} 
  \end{figure}
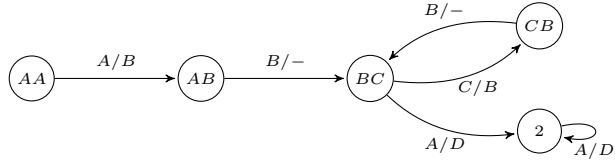

\end{example}

%\noindent \hrulefill

\section{Sketch of proof of Theorem~\ref{theo:exists EN gen}}

Theorem~\ref{theo:exists EN gen} states that in every initialised
 multiplayer cost game that is \propcost\ and \emph{finite-memory}
 coalition-determined, there exists a Nash equilibrium with finite
 memory.

\begin{proof}[Sketch]
 The proof follows the same philosophy than the proof of
 Theorem~\ref{theo:mem lin} and keeps the same notations. Again we
 consider the Nash equilibrium $(\tau_i)_{\ipi}$ defined in the proof
 of Theorem~\ref{theo:exists EN}, whose outcome is $\rho:=\langle
 (\sigmaopt_i)_{i \in \Pi} \rangle_{v_0}$. We recall that for all
 $i \in \Pi$, the strategy $\tau_i$ depends on the strategies
 $\sigmaopt_i$ and $\sigmaopt_{i,j}$ for
 $j \in \Pi \setminus \{i\}$. As the game $\mathcal{G}$ is
 finite-memory coalition-determined by hypothesis, these strategies
 are assumed to be finite-memory. Given $i \in \Pi$ and
 $j \in \Pi \setminus \{i\}$, we denote by $\mathcal{A}^{\sigmaopt_i}$
 (resp. $\mathcal{A}^{\sigmaopt_{i,j}}$) a finite strategy automaton
 for the strategy $\sigmaopt_i$ (resp.  $\sigmaopt_{i,j}$).

 As in the proof of Theorem~\ref{theo:mem lin}, each player needs to
 remember both the play $\rho$ and who has to be punished.  But here
 the play $\rho$ is not anymore the outcome of a positional strategy
 profile: each $\sigmaopt_i$ is a finite-memory
 strategy. Nevertheless, in some sense, we can see the $\sigmaopt_i$'s
 as positional strategies played on the product graph
 $G \times \mathcal{A}^{\sigmaopt_1} \times \cdots \times \mathcal{A}^{\sigmaopt_{|\Pi|}}$. This
 allows us to write $\rho:=v_0\ldots v_{k-1}(v_k\ldots v_n)^{\omega}$
 where\footnote{$|\mathcal{A}|$ denotes the number of states of the
 automaton ${A}$.} $0 \le k \le n \le |V| \cdot \prod_{j\in \Pi}
 |\mathcal{A}^{\sigmaopt_{j}}|$, $v_l \in V$ for all $0 \le l \le
 n$. Like in the proof of Theorem~\ref{theo:mem lin}, we can now
 define, for any $i \in \Pi$, $\mathcal{A}^{\tau_i}$, a finite
 strategy automaton for $\tau_i$. In order to build explicitly
 $\mathcal{A}^{\tau_i}$, we need to take into account, on one hand,
 the path $\rho$, and on the other hand, the memory of the punishing
 strategies $\sigmaopt_{i,j}$. This enables to bound the size of
 $\mathcal{A}^{\tau_i}$ by $ |V| \cdot
\prod_{j\in \Pi} |\mathcal{A}^{\sigmaopt_{j}}|
 + \sum_{j\in \Pi \setminus \{i\}} |\mathcal{A}^{\sigmaopt_{i,j}}|$.
\qed
\end{proof}

%\noindent \hrulefill

\section{Remark on the particular Min-Max cost games of Definition~\ref{def:part cost games}}

\begin{remark}
Note that reachability-price and discounted-price games are
zero-sum\footnote{Let us recall that a Min-Max cost game is zero-sum
  if and only if $\costMin = \costMax$.}  games, whereas the two other
ones are not. For example, let us consider the average-price game
$\mathcal{G}$ depicted on Figure~\ref{fig:AP game}. The vertices of
this game are $A$ and $B$, and the number $0$ or $1$ associated to an
edge corresponds with the price of this edge ($\price(A,B)=
\price(B,B)=1$ and the price of the other edges is zero).
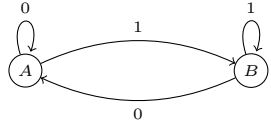
\begin{figure}
  \centering
  \begin{tikzpicture}[xscale=1.5,yscale=1.2]
    \everymath{\scriptstyle}
    
    \path (0,0) node[draw,circle,inner sep=2pt] (q0) {$A$};
    \path (2,0) node[draw,circle,inner sep=2pt] (q00) {$B$};
     
    \path (q0) edge[->,bend left] node[pos=.5,above] {{$1$}} (q00);
    \path (q00) edge[->,bend left] node[pos=.5,below] {{$0$}} (q0);
    \path (q0) edge[->,loop above] node[pos=.5,above] {{$0$}} (q0);
    \path (q00) edge[->,loop above] node[pos=.5,above] {{$1$}} (q00);
      
  \end{tikzpicture}
  
  \caption{Average-price game $\mathcal{G}$.}
  \label{fig:AP game}
\end{figure}

Let $\rho$ be the play $ABAB^2A^2B^4A^4\ldots B^{2^n}A^{2^n}\ldots$,
where $A^i$ means the concatenation of $i$ $A$. Then the sequence of
prices appearing along $\rho$ is
$101^20^21^40^4\ldots1^{2^n}0^{2^n}\ldots$, and so we get:
$\APMin(\rho)=\frac{2}{3}$ and $\APMax(\rho)=\frac{1}{2}$. As these
costs are not equal, the average-price game $\mathcal{G}$ depicted on
Figure~\ref{fig:AP game} is not a zero-sum game. Since an
average-price game is a special case of price-per-reward-average game,
we can conclude that these two kinds of games are non zero-sum games.
\end{remark}

%\noindent \hrulefill

\section{Part of the proof of Theorem~\ref{theo:exists EN PC}}\label{sec:proof CPL}

\begin{proposition}\label{prop:ex nice games}
  Let $\mathcal{G} = (\Pi, V, (V_i)_{i \in \Pi}, E, (\cost_i)_{i \in
    \Pi})$ be a multiplayer cost game where the cost function
  $\cost_i$ belongs to $\{\RPMin,\DPMin,\APMin,\PRAvgMin\}$ for each
  $i \in \Pi$. Then the game~$\mathcal{G}$ is \propcost\ and
  \propcoalpos.
\end{proposition}

\begin{proof}
  Let $\mathcal{G}$ be a a multiplayer cost game where each cost
  function is $\RPMin$, $\DPMin$, $\APMin$ or $\PRAvgMin$. Let us
  first prove that the game $\mathcal{G}$ is \propcost. Given $j \in
  \Pi$, $v \in V$ and $hv \in \hist$, we consider the four possible
  cases for $\cost_j$. Let $\price: E \to \IR$ be a price function and
  $\reward: E \to \IR$ be a diverging reward function.  For the sake
  of simplicity, we write $hv:= h_0\ldots h_k$ with $k \in \IN$,
  $h_k=v$ and $h_l \in V$ for $l=0,\ldots,k$. Moreover, to avoid heavy
  notation, we do not explicitly show the dependency between $\FMin$
  and $j$ in the first case or between $\lambda$ and $j$ in the
  second case.
  \begin{enumerate} 
  
  \item Case $\cost_j=\RPMin$ for a given goal set $\FMin \subseteq
    V$:

    Let us distinguish two situations. If there exists $l \in
    \{0,\ldots,k\}$ such that $h_l \in \FMin$, then we set $a:=
    \sum_{i=1}^{n} \price(h_{i-1},h_i) \in \IR$ and $b:= 0 \in \IR^+$,
    where $n$ is the least index such that $h_n \in \FMin$. Let $\rho$
    be a play with $\first(\rho)=v$, then it implies that
    $\RPMin(h\rho)= \sum_{i=1}^{n} \price(h_{i-1},h_i) = a + b\cdot
    \RPMin(\rho)$ (with the convention that $0 \cdot +\infty=0$).

    If there does not exist $l \in \{0,\ldots,k\}$ such that $h_l \in
    \FMin$, then we set $a:= \sum_{i=1}^{k} \price(h_{i-1},h_i) \in
    \IR$ and $b:= 1 \in \IR^+$. Let $\rho= \rho_0\rho_1\ldots$ be a
    play such that $\first(\rho)=v$.  If $\RPMin(\rho)$ is infinite,
    then $\RPMin(h\rho)=+\infty = a + b \cdot \RPMin(\rho)$.
    Otherwise, if $n$ is the least index in $\IN$ such that $\rho_n
    \in \FMin$, then we have that:
    \[ \begin{array}{ll} 
      \RPMin(h\rho) & = \displaystyle \sum_{i=1}^{k}
      \price(h_{i-1},h_i) + \displaystyle \sum_{i=1}^n
      \price(\rho_{i-1},\rho_i)\\
      & = a + b \cdot \RPMin(\rho).
    \end{array}\]

    \vspace{0.5ex}
  \item Case $\cost_j=\DPMin(\lambda)$ for a given discount factor
    $\lambda \in\, \left]0,1\right[$:
    
    We set $a:= (1-\lambda) \sum_{i=1}^{k} \lambda^{i-1}
    \price(h_{i-1},h_i) \in \IR$ and $b:= \lambda^k \in \IR^+$. Given
    a play $\rho=\rho_0\rho_1\ldots$ such that $\first(\rho)=v$ and
    $\eta:=h\rho \in \plays$ (with $\eta=\eta_0\eta_1\ldots$), we have
    that:
    \[ \begin{array}{ll} 

      \DPMin(\lambda)(h\rho) & = \DPMin(\lambda)(\eta) \\

      & = (1-\lambda)\displaystyle \sum_{i=1}^{+\infty} \lambda^{i-1}
      \price(\eta_{i-1},\eta_i)\\

      & = (1-\lambda)\displaystyle \sum_{i=1}^{k} \lambda^{i-1}
      \price(\eta_{i-1},\eta_i) + (1-\lambda)\displaystyle
      \sum_{i=k+1}^{+\infty} \lambda^{i-1}
      \price(\eta_{i-1},\eta_i)\\ 

      & = (1-\lambda)\displaystyle \sum_{i=1}^{k} \lambda^{i-1}
      \price(h_{i-1},h_i) + \lambda^k
      (1-\lambda)\displaystyle \sum_{i=1}^{+\infty} \lambda^{i-1}
      \price(\rho_{i-1},\rho_{i})\\

      & = a + b\cdot\DPMin(\lambda)(\rho)\,.
    \end{array} \]

    \vspace{0.5ex}
  \item Case $\cost_j=\APMin$:

    We set $a:=0 \in \IR$ and $b:= 1 \in \IR^+$. Given $\rho \in
    \plays$ such that $\first(\rho)=v$ and $\eta:=h\rho \in \plays$
    (with $\eta=\eta_0\eta_1\ldots$), we show that:
    $$\APMin(h\rho)=\APMin(\eta)=\APMin(\rho)\,.$$ If
    $\APMin(\eta)=\APMin(\rho)=+\infty$ or $-\infty$, the desired
    result obviously holds. Otherwise, let us set $x_n:=
    \frac{1}{n}\sum_{i=1}^n \price(\eta_{i-1},\eta_i)$ and $y_n:=
    \frac{1}{n}\sum_{i=1}^n \price(\rho_{i-1},\rho_i)$, for all $n \in
    \IN_0$. By properties of the limit superior and definition of the
    $\APMin$ function, it holds that:
    $$\limsup_{n \to +\infty} (x_n-y_n) \ge \APMin(\eta) -
    \APMin(\rho) \ge \liminf_{n \to +\infty} (x_n-y_n)\,.$$ It remains
    to prove that the sequence $(x_n-y_n)_{n \in \IN}$ converges to
    $0$. For all $n > k$, we have that:
    \[ \begin{array}{ll} 
    |x_n-y_n| & = \left| \frac{1}{n} \cdot \left( \sum_{i=1}^n
    \price(\eta_{i-1},\eta_i) - \sum_{i=k+1}^{k+n}
    \price(\eta_{i-1},\eta_{i})\right)\right|\\
    &  = \frac{1}{n}\cdot\left| \sum_{i=1}^k
    \price(\eta_{i-1},\eta_i) - \sum_{i=n+1}^{n+k}
    \price(\eta_{i-1},\eta_{i}) \right|.
    \end{array} \]
    As the absolute value is bounded independently of $n$ (let us
    remind that $E$ is finite), we can conclude that $(x_n-y_n)_{n \in
      \IN}$ converges to $0$, and so $\APMin(\eta) =\APMin(\rho)$.

    \vspace{0.5ex}
  \item Case $\cost_j=\PRAvgMin$:

    We set $a:=0 \in \IR$ and $b:= 1 \in \IR^+$. Given $\rho \in
    \plays$ such that $\first(\rho)=v$ and $\eta:=h\rho \in \plays$
    (with $\eta=\eta_0\eta_1\ldots$), we show that:
    $$\PRAvgMin(h\rho)=\PRAvgMin(\eta)=\PRAvgMin(\rho)\,.$$ Thanks to
    several properties of $\limsup$, we have that:
    \begin{align}
      \PRAvgMin(\rho) & = \displaystyle \limsup_{n \to +\infty}
      \frac{\sum_{i=1}^n \price(\rho_{i-1},\rho_i)}{\sum_{i=1}^n
        \reward(\rho_{i-1},\rho_i)}\notag\\
      & = \displaystyle \limsup_{n \to +\infty}
      \frac{\sum_{i=1}^n \price(\eta_{k+i-1},\eta_{k+i})}{\sum_{i=1}^n
        \reward(\eta_{k+i-1},\eta_{k+i})}\notag\\
      & = \displaystyle \limsup_{n \to +\infty} \frac{\sum_{i=1}^{n+k}
        \price(\eta_{i-1},\eta_{i}) - \sum_{i=1}^k
        \price(\eta_{i-1},\eta_{i})}{\sum_{i=1}^{n+k}
        \reward(\eta_{i-1},\eta_{i}) - \sum_{i=1}^k
        \reward(\eta_{i-1},\eta_{i})}\notag\\
      & = \displaystyle \limsup_{n \to +\infty} \frac{\sum_{i=1}^{n+k}
        \price(\eta_{i-1},\eta_{i})}{\sum_{i=1}^{n+k}
        \reward(\eta_{i-1},\eta_{i}) - \sum_{i=1}^k
        \reward(\eta_{i-1},\eta_{i})}\label{eq lim 1}\\
      & = \displaystyle \limsup_{n \to +\infty} \frac{\sum_{i=1}^{n+k}
        \price(\eta_{i-1},\eta_{i})}{\sum_{i=1}^{n+k}
        \reward(\eta_{i-1},\eta_{i})}\cdot  \frac{1}{1 - \frac{\sum_{i=1}^k
        \reward(\eta_{i-1},\eta_{i})}{\sum_{i=1}^{n+k}
        \reward(\eta_{i-1},\eta_{i})}}\notag\\
      & = \displaystyle \limsup_{n \to +\infty} \frac{\sum_{i=1}^{n+k}
        \price(\eta_{i-1},\eta_{i})}{\sum_{i=1}^{n+k}
        \reward(\eta_{i-1},\eta_{i})}\label{eq lim 2}\\
      & = \displaystyle \limsup_{n \to +\infty} \frac{\sum_{i=1}^{n}
        \price(\eta_{i-1},\eta_{i})}{\sum_{i=1}^{n}
        \reward(\eta_{i-1},\eta_{i})}\notag\\
      & = \PRAvgMin(\eta) = \PRAvgMin(h\rho)\,.\notag
    \end{align} 
    Line~\eqref{eq lim 1} comes from the fact that the reward function
    $\reward$ is diverging, and from the following property: if
    $\lim_{n \to +\infty} b_n = b \in \IR$, then $\limsup_{n \to
      +\infty} (a_n + b_n) = (\limsup_{n \to +\infty} a_n) +b$.
    Line~\eqref{eq lim 2} is implied by this property: if $\lim_{n \to
      +\infty} b_n = b > 0$, then $\limsup_{n \to +\infty} (a_n \cdot
    b_n) = (\limsup_{n \to +\infty} a_n) \cdot b$.
  \end{enumerate}
   Note that, if the history $h$ is empty, then $k=0$ and, in all
  cases, $a$ is equal to $0$ and $b$ to 1. This actually implies that
  $\cost_i(h\rho)=\cost_i(\rho)$ holds.
 
  Let us now prove that the game $\mathcal{G}$ is \propcoalpos. Given
  a player $i \in \Pi$, if $\cost_i=\RPMin$, then we take
  $\costMax^i=\RPMax$. We do the same for the other cases by defining
  the gain function $\costMax^i$ for the coalition as the counterpart of
  $\cost_i$ in Definition~\ref{def:part cost games}.  Clearly, it
  holds that $\cost_i \ge \costMax^i$. Moreover, the Min-Max cost game
  $\mathcal{G}^i=(V,V_i,V \setminus V_i,E,\cost_i,\costMax)$ is
  determined and has positional optimal strategies by
  Theorem~\ref{theo:determined}. \qed
\end{proof}

\end{document}